\documentclass[twoside,11pt]{article}

\usepackage[abbrvbib, preprint]{jmlr2e}

\usepackage{xcolor}
\usepackage{colortbl}

\usepackage{floatrow}
\newfloatcommand{capbtabbox}{table}[][\FBwidth]

\usepackage{comment}
\usepackage{booktabs}       

\usepackage{tikz} 
\usepackage{url}
\usepackage{hyperref}
\usepackage{graphicx}
\usepackage{dsfont}
\usepackage{amsmath}
\usepackage{amssymb}
\usepackage{graphics}

\usepackage{empheq}

\usepackage{algorithm}
\usepackage{algorithmicx}
\usepackage{algpseudocode}

\algnewcommand\algorithmicinput{\textbf{Input:}}
\algnewcommand\algorithmicoutput{\textbf{Output:}}
\algnewcommand\Input{\item[\algorithmicinput]}%
\algnewcommand\Output{\item[\algorithmicoutput]}%
\newcommand{\algrule}[1][.2pt]{\par\vskip.5\baselineskip\hrule height #1\par\vskip.5\baselineskip}
\newcommand{\thickrule}[1][.5pt]{\par\vskip.5\baselineskip\hrule height #1\par\vskip.5\baselineskip}

\usepackage{graphicx}
\usepackage{subcaption}
\usepackage{listings}

\usepackage{bbm}
\usepackage{multirow}
\usepackage{rotating}
\usepackage{multicol}

\def\EE{\mathbb{E}}

\usepackage{tcolorbox}
\usepackage{multirow}

\newtcolorbox{mybox}[3][]
{
  colframe = #2!80,
  colback  = #2!10,
  coltitle = #2!10!black,
  title    = {#3},
  boxsep   = 0.25pt,
  left     = 0.5pt,
  right    = 0.5pt,
  top      = 0pt,
  bottom   = 0pt,
  width=\linewidth,
  #1,
}
\DeclareCaptionType{alg}[\textbf{Algorithm}][]

\usepackage{url}

\usepackage{tikz}

\definecolor{Maroon}{RGB}{204, 102, 0}

\definecolor{rulecolor}{rgb}{0.0, 0.06, 0.54}
\definecolor{tableheadcolor}{rgb}{0.88, 0.94, 0.87}
\definecolor{orange}{rgb}{1, 0.49, 0}
\definecolor{bluecolor}{rgb}{0.74, 0.83, 0.9}

\begin{document}

\title{Radial Spike and Slab Bayesian Neural Networks for Sparse Data in Ransomware Attacks}

\author{\name Jurijs Nazarovs\thanks{Work performed at Microsoft.} \email nazarovs@wisc.edu \\
       \addr Department of Statistics\\
       University of Wisconsin\\
       Madison, WI 53706, USA
       \AND
       \name Jack W.\ Stokes \email jstokes@microsoft.com \\
       \addr Microsoft\\
       Redmond, WA 98052, USA
       \AND
       \name Melissa Turcotte{\footnotemark[1]}\email melissa.turcotte@gmail.com \\
       \addr Securonix\\
       Santa Fe, NM 87506, USA
       \AND
       \name Justin Carroll \email justin.carroll@microsoft.com \\
       \addr Microsoft\\
       Redmond, WA 98052, USA
       \AND
       \name Itai Grady \email igrady@microsoft.com \\
       \addr Microsoft\\
       Redmond, WA 98052, USA       \\
       }

\maketitle

Ransomware attacks
are increasing at an alarming rate, leading to large financial losses, unrecoverable encrypted data, data leakage, and privacy concerns. The prompt detection of ransomware attacks is required to minimize further damage, particularly during the encryption stage. However, the frequency and structure of the observed ransomware attack data makes this task difficult to accomplish in practice. The data corresponding to ransomware attacks represents temporal, high-dimensional sparse signals, with limited records and very imbalanced classes. While traditional deep learning models have been able to achieve state-of-the-art results in a wide variety of domains, Bayesian Neural Networks, which are a class of probabilistic models, are better suited to the issues of the ransomware data. These models combine ideas from Bayesian statistics with the rich expressive power of neural networks. In this paper, we propose the Radial Spike and Slab Bayesian Neural Network, which is a new type of Bayesian Neural network that includes a new form of the approximate posterior distribution. The model scales well to large architectures and recovers the sparse structure of target functions. We provide a theoretical justification for using this type of distribution, as well as a computationally efficient method to perform variational inference. We demonstrate the performance of our model on a real dataset of ransomware attacks and show improvement over a large number of baselines, including state-of-the-art models such as Neural ODEs (ordinary differential equations). In addition, we propose to represent low-level events as MITRE ATT\&CK tactics, techniques, and procedures (TTPs) which allows the model to better generalize to unseen ransomware attacks.


\section{Introduction}
Ransomware attacks are increasing rapidly and causing significant losses to governments, corporations, non-governmental organizations, and individuals.
The losses may include financial costs due to ransoms paid to decrypt assets, unrecoverable files
when the ransom is not paid or the attacker fails to provide the decryption key, privacy and intellectual property theft when assets are exported, and even significant injury when ransomware impairs health care devices or patient records in hospitals.
It is clear that the timely detection of ransomware incidents is necessary in order to minimize the number of assets that are encrypted or exfiltrated~\cite{9392548}. To improve the ransomware response,
this work proposes a new Bayesian Neural Network model that offers improved detection rates for organizations which employ analysts to protect
their assets and networks.

The problem is usually considered as a detection task, where the two classes are ransomware or not.
The traditional methods of statistics and machine learning have been proposed to detect security threats in general and specifically ransomware in some cases.
From the statistical perspective, a common approach is the application of Bayesian Networks
~\cite{perusquia2020bayesian,oyen2016bayesian,shin2015development}, whose main goal is to model the relationship between the observed signal and the class of the attack as a graphical model.
From the machine learning perspective, a range of models were used to detect ransomware~\cite{alhawi2018leveraging,poudyal2018framework,zhang2019classification,larsen2021survey}, such as Naive Bayes, Gradient Boosting, and Random Forests.

{\bf Bottleneck.} To obtain the rich expressive power of traditional deep learning models, training usually requires having access to a large number of records to successfully obtain robust generalized results.
Unfortunately, the frequency and structure of commonly observed data corresponding to ransomware attacks makes this task more difficult to accomplish. In particular, ransomware attack data can be represented as temporal high-dimensional sparse signals, with a limited number of records and very imbalanced classes. In our data, the percentage of ransomware attacks to non-ransomware attacks is 1\% versus 99\%, respectively.

{\bf Main ideas and contributions.} To address these unique features of the ransomware data, we first propose to represent ransomware signals according
their MITRE ATT\&CK tactics, techniques, and procedures (TTPs) which allows us to generalize ransomware and other attacks at a higher-level instead of the low-level detections associated with an individual attack. In addition, this allows for the detection of both human operated and automated ransomware attacks across multiple stages in the
kill chain within an organization's network.
Next,
we propose a new probabilistic model which is called the Radial Spike and Slab Bayesian Neural Network. It is a Bayesian Neural Network, where the approximate posterior is represented by a mixture of distributions, resulting in a Radial Spike and Slab distribution. Our model provides the following benefits including:
\begin{itemize}
\item the Spike and Slab component handles missing or sparse data,
\item the Radial component scales well with the growth of the number of parameters in the deep neural network, and
\item the Bayesian component prevents overfitting in the limited data setup. 
\end{itemize}
From the theoretical perspective, we provide the justification for using this type of distribution, as well as a computationally efficient method to perform variational inference.
In the results section, we demonstrate the performance of our model on a set of actual ransomware attacks and show improvement over a number of baselines, including the state-of-the-art temporal models such as RNNs~\cite{cho2014learning} and Neural ODEs (ordinary differential equations)~\cite{chen2018neural}. Thus, the proposed model
is an important tool for the critical problem of ransomware detection. 
\section{Incident Data Description}
\label{sec:data}
This work utilizes threat data provided by
Microsoft Defender for Endpoint
and Microsoft Threat Experts, Microsoft’s
managed threat hunting service NeurIPS
to detect ransomware and other types of cybersecurity attacks.
Low-level event generators are manually created by analysts (i.e., rules) and are provided with a UUID.

\textbf{Features.} Given each incident, features need to be extracted which capture the range of attack behaviors observed across the kill chain and represent common behaviors across the different families of ransomware attacks. The low-level events cannot be used directly because there are too many to train our model, given the number of labeled examples, and they do not generalize well individually. To overcome
these problems, we map a subset of the low-level events into a higher-level representation using the MITRE ATT\&CK framework~\cite{mitreattack}.  We chose the MITRE ATT\&CK framework for the mapping because it provides a knowledge base of adversary tactics, techniques, and procedures (TTPs) and is widely used across the industry for classifying attack behaviors and understanding the lifecycle of an attack. Using the MITRE ATT\&CK TTPs is a natural choice for features as it is generalizable, interpretable, and easy to acquire for this data as each low-level event from
Microsoft
is tagged with the MITRE technique associated with the alerted behavior~\cite{mitreattack}.
For example, one of the features can represent whether `OS Credential Dumping' happened or not. Additional MITRE ATT\&CK features are included in Table~\ref{tab:import},
and the entire set is provided by the MITRE corporation~\cite{mitre}.
In total, our data is a sparse high-dimensional vector of size 706, which contains 298 MITRE ATT\&CK features and 408 additional rule-based features, at each time point.
However, one of the characteristics of the data is sparsity because only very few actions are completed at each time step during the attack.

\textbf{Labels.} Using manual investigation, an analyst provides a label for each incident indicating whether it is due to a ransomware attack or another type of attack.
The ransomware incidents include both
human operated ransomware (HumOR) and automated ransomware attacks.
However, our positive class label only indicates that an attack is ransomware and does not distinguish between the two classes of ransomware (i.e., HumOR, Automated).
Our goal is to build an alarm-recommendation system, which can not only detect a possible ransomware attack, but also provide an estimate of the uncertainty about the decision. We provide additional details about the training and testing data in Section~\ref{sec:experiment}.

\section{Methodology}
Important features of probabilistic models, such as providing a notion of uncertainty, dealing with missing data, and preventing overfitting in a limited data regime, have generated a strong interest in deep Bayesian learning.
In this section, we provide more details regarding Bayesian Neural Networks, including different aspects of initializing and training the model. We then propose the Radial Spike and Slab Bayesian Neural Network model to address the problems of the ransomware data.

\textbf{Bayesian Neural Networks.}
The main idea behind the Bayesian Neural Network is to consider all weights as being samples from a random distribution.
Formally, we denote the observed data as $(x, y)$, where $x$ is an input to the network, and $y$ is a corresponding response.
Let all weights of a BNN, $W=(W^1, \ldots, W^D)$, be a random vector, where $D$ is the depth (i.e., number of layers) of the BNN and each $W^j = (w^{j, 1}, \ldots, w^{j, l_j})$ is a random vector itself of all weights $w^{j,k}$ per layer $W^j$ of size $l_j$.
To generate uncertainty of the prediction, we need to be able to compute $p(y|x)$. However, since all weights of a BNN are considered to be random variables, we can rewrite the conditional probability as
$p(y|x) = \int_w p(y,W|x)dW = \int_W p(y|W, x) p(W|x)dW$. Typically, the likelihood term $p(y|W, x)$ is defined by the problem setup, e.g., if we consider classification, as in ransomware incident detection, $y \sim Bern(g(W, x))$ for some function $g$. Then, the main problem of training a BNN is to compute the posterior probability $p(W|x)$, given the observed data $x$ and a suitable prior for $W$.

In some simple cases of small neural networks, it may be possible to obtain a closed-form solution for the posterior if the prior and posterior are conjugate distributions.
In other cases, if a closed-form solution is unavailable, sampling-based strategies are required such as Markov Chain Monte Carlo schemes based on Gibbs or Metropolis Hasting samplers.
While such an approach provides excellent statistical behavior with theoretical support, scalability as a function of the dimensionality of the problem is known to be a serious issue.
The alternative for machine learning and vision problems is Variational Inference (VI)~\cite{graves2011practical}.
The core concept of VI is based on the fact that approximating the true posterior with another distribution  may often be acceptable in practice.
The computational advantages of VI permit estimation procedures
in cases which would not otherwise be feasible.
VI is now a mature technology, and its success has led to a number of follow-up developments focused on
theoretical as well as practical aspects~\cite{blundell2015weight}.

When using VI in Bayesian Neural Networks,
we approximate the true unknown posterior distribution $P(W|x)$ with an \textit{approximate posterior} distribution $Q_\theta$ of \textit{our choice}, which depends on learned parameters $\theta$. Let $W_\theta=(W_\theta^{1}, \ldots, W_\theta^{D})$ denote a random vector with distribution $Q_\theta$ and probability distribution function (pdf) $q_\theta$.
VI seeks to find $\theta$ such that $Q_\theta$ is as close as possible to the real (unknown) posterior $P(W|x)$, and this is accomplished by minimizing the Kullback–Leibler ($KL$) divergence between $Q_\theta$ and $P(W|x)$. Given a prior pdf of weights, $p$, with a likelihood term $p(y|W, x)$, and the common \textit{mean field} assumption of independence for $W^d$ and $W_\theta^d$ for $d \in 1, \ldots, D$, i.e., $p(W) = \prod_{d=1}^D p^d (W^d)$ and $q_\theta (W_\theta) = \prod_{d=1}^D q_\theta^d (W_\theta^d)$, 

\begin{align}
    &\boldsymbol{\theta}^{*} = \underset{\theta}{\arg \min}\
    KL\left(q_\theta||p\right) -
    \EE_{q_{\theta}}\left[\ln p(y| W, x)\right]\label{eq:vi}
\end{align}
\begin{align}
    &KL\left(q_\theta||p\right) =\sum_{d=1}^D\EE_{q_{\theta}^d}\left[ \ln q_{\theta}^d(w)\right] -
        \EE_{q_{\theta}^d}\left[\ln p^d(w)\right].
        \label{eq:kl}
\end{align}

By definition of the expected value $\EE_{q_\theta}$, it is necessary to compute the multi-dimensional integral w.r.t $w\sim Q_\theta$ to solve \eqref{eq:vi}.
If such integrals are impossible to compute in a closed-form, a numerical approximation is used \cite{ranganath2014black,paisley2012variational,miller2017reducing}.
For example, Monte Carlo (MC) sampling yields an asymptotically exact, unbiased estimator with variance $O(\frac{1}{M})$, where $M$ is the number of samples. For a function $g(\cdot)$:
\begin{align}
    \mathbb{E}_{q_\theta}\left[g(w)\right]
    =
    \int{g(w)q_\theta(w)dw}
    \approx \frac{1}{M}\sum_{i=1}^Mg(w_i),
    \text{ where } w_i \sim Q_\theta.
\label{eq:mc-intro}
\end{align}
The expected value terms in \eqref{eq:vi} and \eqref{eq:kl} can be estimated by applying the method in \eqref{eq:mc-intro}, and in fact, even if a closed-form expression can be computed, an MC approximation may perform similarly given enough samples~\cite{blundell2015weight}.

Given a mechanism to solve \eqref{eq:vi}, the main consideration in VI is the \textit{choice of prior $p$ and the approximate posterior $q_\theta$}.
A common choice for $p$ and $q_\theta$ is Gaussian, which allows calculating \eqref{eq:kl} in a closed-form. However, this type of distribution is mainly used for computational purposes and does not reflect the nature of the data.
Choosing the correct distribution, especially the one which can incorporate the features of the analyzed data, is an open problem~\cite{ghosh2017model,farquhar2019radial,mcgregor2019stabilising,krishnan2019efficient}. In the next section, we discuss our proposed distribution, which naturally fits the data encountered in ransomware incident detection.
While we give the description of the analyzed data in Section~\ref{sec:data}, we next describe the features of the data, which are important to encapsulate in the model design.

\textbf{Spike and Slab distribution.}
The sparsity of the data
is a common problem in many areas~\cite{kang2013prevention} and was previously approached from different perspectives.
For example in the statistics community, sparsity can be addressed with both Stochastic Regression Imputation and Likelihood Based Approaches~\cite{lakshminarayan1999imputation}.
In the machine learning community,  methods based on k-nearest neighbor~\cite{batista2003study} and iterative techniques \cite{buuren2010mice} have been developed, including approaches with neural networks~\cite{sharpe1995dealing,smieja2018processing}.
Another way to tackle sparsity comes from regularization theory via L1 regularization, e.g., group LASSO~\cite{meier2008group}, sparse group LASSO~\cite{simon2013sparse} and graph LASSO~\cite{jacob2009group}.

However, we are interested in a probabilistic approach to address the sparsity in our data. From the probabilistic perspective, a common way to account for sparsity of the data in the model is to  consider an appropriate distribution. For example, the distribution can be the Horseshoe distribution~\cite{carvalho2009handling} or derivatives of the Laplace distribution~\cite{babacan2009bayesian,bhattacharya2015dirichlet}.
Another common way is, instead of one distribution, to consider  the mixture of priors with Spike and Slab components which have been widely used for Bayesian variable selection~\cite{mitchell1988bayesian,george1997approaches}.
In general, the form of the Spike and Slab distribution for random variable $w$ can be written as:
$w \sim (1-\pi) \delta_{\xi}+\pi g$,
where $\pi$ determines the probability for each mixture component, $\delta$ is spike component, which is modeled with a Dirac delta  function such that
$
\delta(w)= \begin{cases}+\infty, & w=\xi \\ 0, & w \neq \xi \end{cases}
$
and
$
\int_{-\infty}^{\infty} \delta(w) d w=1
$,
and $g$ is the slab  component, which is a general distribution of the practioner's choice.
The general idea is to explicitly introduce the sparsity component in the distribution of the data, allowing the probability mass to fully concentrate on $\xi=0$ with probability $1-\pi$, and with probability $\pi$ spread the remaining mass over the domain of the slab  component $g$.
Notice, that $\pi$ can be considered as a random variable itself, e.g., $\pi \sim Bern(\lambda)$, where $\lambda$ is either a learned parameter or a fixed value that is provided by a specialist.

The next questions are: (1) how can the `Spike and Slab' distribution be applied in a BNN, and (2) which slab  component $g$ should we consider?

\textbf{Spike and Slab BNN.}
In the BNN, all of the neural network's weights $W$ are considered to be random variables, and to use VI to solve \eqref{eq:vi}, for each layer's set of weights $W^j$ in $W=(W^1, \ldots, W^D)$, it is necessary to provide the prior $p^j$ and the approximate posterior $q_\theta^j$.
Without loss of generality, we consider a single weight $w:=w^{j, k}$, dropping the indices $j$ and $k$, and only work with the prior $p$ and the approximate posterior $q$ for the remainder of this section.

Incorporating a Spike and Slab distribution on both the prior $p$ and the approximate posterior $q$, samples $w_p$ from $p$ and $w_q$ from $q$ have the following distribution:
\begin{equation}
\small
w_p|\pi_p \sim (1-\pi_p) \delta_0+\pi_p g_p
\text{ and }
w_q|\pi_q \sim (1-\pi_q) \delta_0+\pi_q g_q,
\label{eq:ss_form}
\end{equation}
where $\pi_p \sim Bern(\lambda_p), \pi_q \sim Bern(\lambda_q)$, and $g_p, g_q$ are distributions of our choice.

As we discussed previously, the main goal of VI is to learn parameters $\theta$ of an approximate posterior $q_\theta$, by minimizing \eqref{eq:kl}.
In the case of \eqref{eq:ss_form},
$\theta = (\lambda_q, \theta_q)$,
where $\lambda_q$ is the probability of the Bernoulli distribution associated with $\pi_q$, and $\theta_q$ are the parameters of the Slab component $g_q$.
First, we state
Theorem~\ref{thm:klmix}, which allows us to compute the $KL$ term between two general Spike and Slab distributions.
\begin{theorem}
    Given two general Spike and Slab distributions such that:
    $p(w|\pi_p)=(1-\pi_p) \delta_0(w)+\pi_p g_p(w),$
    $q(w|\pi_q)=(1-\pi_q) \delta_0(w)+\pi_q g_q(w),$
    $\pi_p \sim p(\pi) = Bern(\lambda_p)$, and
    $\pi_q \sim q(\pi) = Bern(\lambda_q)$,
    with $\delta_0$ being a dirac delta function at 0 and
    $g_p, g_q$ are the pdfs of the distributions of our choice,
     the $KL\left(q(w,\pi)\|p(w,\pi)\right)$ is equal to:
     \begin{equation}
         KL\left(Bern(\lambda_q)\|Bern(\lambda_p)\right) + \lambda_q KL\left(g_q\|g_p\right).
         \label{eq:klmix}
     \end{equation}
     \label{thm:klmix}
\end{theorem}
\begin{proof}
    \begin{align*}
        &KL\left(q(w,\pi)\|p(w,\pi)\right) \\
        =~& \int_\pi \int_w \log\frac{q(w, \pi)}{p(w, \pi)}q(w, \pi)dw d\pi \\
        \text{given} & \text{ that } q(w, \pi)=q(w|\pi)q(\pi) \text{ and } p(w, \pi)=p(w|\pi)p(\pi)\\
        =~&\int_\pi \left\{\int_w \log\frac{q(w, \pi)}{p(w, \pi)}q(w|\pi)dw\right\} q(\pi)d\pi\\
        \text{given} & \text{ that }  q(\pi)=Bern(\lambda_q) \text{ and } p(\pi)=Bern(\lambda_p)\\
        =~& q(\pi=0)\left\{\int_w \log\frac{q(w|0)q(\pi=0)}{p(w|0)p(\pi=0)}q(w|0)dw\right\} \\
        &+q(\pi=1)\left\{\int_w \log\frac{q(w| 1)q(\pi=1)}{p(w| 1)p(\pi=1)}q(w|1)dw\right\} \\
        =~& (1-\lambda_q)\left\{\log\frac{1-\lambda_q}{1-\lambda_p}\int_w\delta_0(w)dw\right\}\\
        &+\lambda_q\left\{\log\frac{\lambda_q}{\lambda_p} + \int_w\log\frac{g_q(w)}{g_p(w)}g_q(w)dw\right\}\\
        =~& (1-\lambda_q)\log\frac{1-\lambda_q}{1-\lambda_p} + \lambda_q\log\frac{\lambda_q}{\lambda_p}\\
        &+\lambda_q \int_w\log\frac{g_q(w)}{g_p(w)}g_q(w)dw\\
        =~& KL\left(Bern(\lambda_q)\|Bern(\lambda_p)\right) + \lambda_q KL\left(g_q\|g_p\right).
    \end{align*}
\end{proof} 
\textbf{Choice of $g_q$ and $g_p$: Radial distribution.}
So far, we have shown results for a general Spike and Slab distribution. The important question is which slab  components $g$ should we consider for our approach, and if $g_q$ and $g_p$ should be from the same family?
Authors in~\cite{bai2020efficient} considered both $g_q$ and $g_p$ to be the Gaussian distribution.
However, there is emerging evidence~\cite{ farquhar2019radial, fortuin2020bayesian} that the Gaussian assumption results in poor performance of the medium to large-scale Bayesian Neural Networks.  Authors regard this as being caused by the probability mass in a high-dimensional Gaussian distribution concentrating in a narrow ``soap-bubble'' far from the mean. For this reason,~\cite{farquhar2019radial} proposed using a Radial distribution with parameters ($\mu$, $\sigma$), where samples can be generated as:
 \begin{equation}
      \mu + \sigma *\frac{\xi}{||\xi||}*|r| \sim Radial(\mu, \sigma), where \xi \sim MVN(0, I), r \sim N(0,1).
     \label{eq:samprad}
 \end{equation}

Following~\cite{farquhar2019radial}, we set up our approximate posterior $g_q$ to be the Radial distribution ($\mu$, $\sigma$), while the prior $g_p$ is Normal(0, 1).
Given equation \eqref{eq:klmix}, it is necessary to define the $KL\left(g_q\|g_p\right)$ term. Unfortunately, a closed-form solution for our choice of $g_q$ and $g_p$ is not available, and we approximate the $KL$ term using Monte Carlo sampling from equation \eqref{eq:mc-intro} with $M$ samples. This process leads to (up to a constant):
$KL\left(g_q\|g_p\right)\approx
- \log \sigma
-\frac{1}{M}\sum_{i=1}^M\log [p(w_i)]$,
where $w_i$ is sampled from the Radial distribution ($\mu$, $\sigma$) as described in equation \eqref{eq:samprad} and $p$ is the Likelihood of $N(0, 1)$.
Note that running an MC approximation for large $M$ can lead to running out of memory in either a GPU or RAM,~\cite{nazarovs2021graph}. To tackle this issue, we follow~\cite{nazarovs2021graph} and apply a graph parameterization for our Radial Spike and Slab distribution, allowing us to set $M=1000$ without exhausting the memory.

\textbf{Reparameterization trick: Gumbel-Softmax.}
Given Theorem~\ref{thm:klmix}, we can rewrite equation \eqref{eq:vi} as:
\begin{align}
    \boldsymbol{\theta}^{*} = \underset{\theta=(\lambda_q, \theta_q)}{\arg \min}\
    KL\left(Bern(\lambda_q)\|Bern(\lambda_p)\right) + \lambda_q KL\left(g_q\|g_p\right)
    -\EE_{q_{\theta}}\left[\ln p(y| W, x)\right].
\label{eq:vimix}
\end{align}
Recall, we can compute the $KL\left(Bern(\lambda_q)\|Bern(\lambda_p)\right)$ in a closed-form (inside the proof of Theorem~\ref{thm:klmix}) and
approximate the $KL\left(g_q\|g_p\right)$ term with MC sampling.
Next, there are two main aspects left for our attention:
(1) computing $\EE_{q_{\theta}}\left[\ln p(y| W, x)\right]$, which is usually approximated with Monte-Carlo sampling~\cite{kingma2013auto} because of the intractability issue, and
(2) how to do back-propagation for optimization.
The problem with back-propagation in this setting is that sampling directly from, e.g., $w\sim N(\mu,\sigma)$ with learnable parameters $\mu$ and $\sigma$, does not allow us to back-propagate through those parameters, and thus, they cannot be learned.
This issue is addressed by applying a local-reparameterization trick~\cite{kingma2015variational}.
For example, instead of sampling from $w\sim N(\mu, \sigma)$ directly, we sample  $z\sim N(0,1)$ and compute: $w=\mu+\sigma z$. This allows back-propagation to optimize the loss w.r.t. $\mu$ and $\sigma$.

While the local-reparameterization trick is obvious for members of a location-scale family, like the Gaussian distribution, and even for the selected Radial distribution, it is not clear how to apply this trick to the Bernoulli distribution, $Bern(\lambda)$.  
One way to address this issue is to approximate samples from the Bernoulli distribution with the Gumbel-Softmax~\cite{maddison2016concrete,jang2016categorical,bai2020efficient}.
That is, $\pi \sim Bern\left(\lambda\right)$
is approximated by
$\widetilde{\pi} \sim \text{Gumbel-Softmax}\left(\lambda, \tau\right)$, where
$\widetilde{\pi}=\left(1+\exp \left(-\eta / \tau\right)\right)^{-1}$,
$\eta =\log \frac{\lambda}{1-\lambda}+\log \frac{u}{1-u}$,
and $u \sim \mathcal{U}(0,1)$.
Here, $\tau$ is the parameter which is referred as the temperature. When $\tau$ approaches $0, \tilde{\pi}$ converges in distribution to $\pi$. However, in practice, $\tau$ is usually chosen no smaller than $0.5$ for numerical stability~\cite{bai2020efficient}.
Applying the Gumbel-Softmax approximation instead of optimizing the loss for parameter $\lambda_q$, we consider a new parameter
$\theta_\pi = \log{\frac{\lambda_q}{1-\lambda_q}}$.
Thus, $\lambda_q=S(\theta_\pi)=\frac{1}{1+e^{-\theta_\pi}}$,
resulting in the final learned parameters: $\theta=(\theta_\pi, \theta_q)$.

\textbf{Final Loss and Method Summary.}
A step-by-step summary of the method in provided in Algorithm~\ref{algo:method}. The final loss is given in Algorithm~\ref{algo:loss}.
\begin{mybox}[float=ht!]{gray}{}
\begin{center}

\captionof{alg}{Learning the posterior distribution of a BNN $p(W|x)$ with a Radial Spike and Slab approximate posterior, to account for sparsity of the data. \label{algo:method}
}
\vspace{-15pt}
\thickrule

\begin{algorithmic}[1]
    \footnotesize
    \Input
    \\
    Neural Network of depth $D$ with
    \\
    Weights $W_\theta=(W_\theta^{1}, \ldots, W_\theta^{D})$, which have
    \\
    Spike and Slab Radial distribution $Q_\theta$ with pdf $q_\theta$, s.t.
    \begin{itemize}
        \item $q(w|\pi_q)=(1-\pi_q) \delta_0(w)+\pi_q g_q(w; \mu, \sigma),$
        \item $g_q(w; \mu, \sigma)$ is pdf of $Radial(\mu,\sigma)$
        \item $\pi_q \sim Bern(S(\theta_\pi))$,  where $S$ is the softmax, and
    \end{itemize}
    \\
    Prior Spike and Slab distribution $P_\theta$ with pdf $p$, s.t.
    \begin{itemize}
        \item $p(w|\pi_p)=(1-\pi_p) \delta_0(w)+\pi_p g_p(w; \mu_p, \sigma_p),$
        \item $g_p(w; \mu_p, \sigma_p)$ is pdf of Gaussian distribution
        \item $\pi_p \sim Bern(\pi_p)$
    \end{itemize}

    \Output{Learned parameters $\theta=(\theta_\pi, \mu, \sigma)$}
    \algrule

    \Require  Prior distribution's parameters $(\pi_p, \mu_p, \sigma_p)$
    \algrule

    \While{$\theta$ has not converged}
        \State Minimize VI loss in equation~\eqref{eq:vifinal}, by using
        gradient descent algorithms (e.g., SGD or Adam) and doing:
        \State \textbf{Forward pass}: to compute
        \begin{itemize}
            \item $y$ with local reparameterization trick for both Radial and Bernoulli (using Gumbel-Softmax)
            \item $KL$ terms and expected log-likelihood term, using combination of closed-form and MC
        \end{itemize}
        \State \textbf{Backward pass}: compute gradients of $\theta$
    \EndWhile
\end{algorithmic}
\end{center}
\end{mybox}


\begin{mybox}[float=ht!]{gray}{}
\footnotesize
\begin{center}
\captionsetup{labelformat=empty}
\captionof{alg}{\textbf{Algorithm 2:} Final loss used for optimization in Algorithm~\ref{algo:method}.
\label{algo:loss}
}
\vspace{-15pt}
\thickrule

{\begin{flushleft} Original:\end{flushleft}\vspace{-30pt}}
\begin{align*}
        &KL\left(Bern(\lambda_q)\|Bern(\lambda_p)\right) +
    \lambda_q KL\left(g_q\|g_p\right)
    -\EE_{Q_{\theta}}\left[\ln p(y| W, x)\right]\nonumber
\end{align*}
\vspace{-20pt}
\algrule
{\begin{flushleft} Final:\end{flushleft}\vspace{-30pt}}
\begin{align}
L = \sum_{\substack{j=1,\ldots,D,\\ k=1,\ldots, l_j}}{KL_{j, k}}-\EE_{Q_{\theta}}\left[\ln p(y| W, x)\right], \text{ where}
\label{eq:vifinal}
\end{align}
 $   KL_{j, k}=(1-S(\theta_\pi^{j, k}))\log\frac{1-S(\theta_\pi^{j, k})}{1-\lambda_p^{j, k}} + S(\theta_\pi^{j, k})\log\frac{S(\theta_\pi^{j, k})}{\lambda_p^{j, k}}
        +S(\theta_\pi^{j, k})\left\{- \log \sigma^{j, k}
-\frac{1}{M}\sum_{i=1}^M\log [p(\mathbf{w_i^{j, k}})]\right\}$
\label{eq:kl_layer}
\algrule
{\begin{flushleft}
Note that based on the \textit{mean field} assumption of a BNN, the final loss $L$ includes the sum over all $KL_{j, k}$ terms, which are computed
for each $k$-th weight $w^{j,k}$ of the $j$-th layer of the BNN with parameters $\theta^{j,k} = (\theta_\pi^{j,k}, \mu^{j,k}, \sigma^{j,k})$.
In this case, the final set of trainable parameters is $\theta=\{\theta^{j,k}\}$ for $j=1,\ldots, D$ and $k=1,\ldots, l_j$.
In addition, $\EE_{Q_{\theta}}$ can be computed either in a closed-form or approximated by MC, depending on the complexity of the BNN.
\end{flushleft}}
\end{center}
\end{mybox}

\section{Experiments}
\label{sec:experiment}

\textbf{Data description.}
As described previously in Section~\ref{sec:data}, each incident is represented by a temporal sequence of events from a
knowledge base of TTPs with an
assigned label, which indicates whether it is ransomware or another type of attack.
First, the company provided 201 incidents
labeled as Ransomware and 24,913 with Non-Ransomware labels for the initial dataset.
All of the samples in this dataset were deduplicated and included 706 sparse binary features.
This first dataset was randomly split with 80\% of the examples assigned to the
training set, while the remainder were used to create a validation
set.
Second, for the test set, we received a newer, deduplicated dataset
making it independent of the training and validation sets.
This dataset included
644 Ransomware incidents and 14,696
Non-Ransomware incidents.

\textbf{Preprocessing of temporal information.}
Some of the models such the Neural ODE benefit from knowledge of the actual time associated with the recorded event, while others, including
the RNN with a GRU cell,
can be trained on the event sequence based solely on the event index (i.e., t=1,2,...). Finally, other models such as the fully connected and Bayesian Neural Networks
can be trained and tested using the aggregation of all of the events in the event sequence.
To reduce the number of time steps for the time-based models for our study, we aggregated all TTP events observed within a one minute window.
We set the aggregation time to one minute after doing hyperparameter tuning on this value.
This results in very few signals being recorded per aggregated time step.
We see that the majority of the data have a small number of features that are set, namely less than 10 out of 706 possible.
For the neural network models, we aggregated all of the TTP features into a single input vector. All of the sequences for the training and testing datasets
were truncated after one hour from the time of the first event.

\textbf{Models.}
In the experiments, we consider several baseline models from the traditional, temporal, and probabilistic deep learning settings, in addition
to our proposed model.
From the \textit{temporal perspective}, we consider two models including the Recurrent Neural
Network with a GRU cell (RNN) and the Neural ODE (NODE).
As we mentioned earlier, the traditional
recurrent neural network models (e.g., Simple RNN, GRU, LSTM) ignore
the value of the time steps and only consider the order (i.e., index), in contrast
to the Neural ODE which accounts for the time step value.
\textit{Note}, we originally considered several temporal models, which do not account
for the time value, like the traditional (i.e., Simple) RNN, the RNN with a GRU cell, the LSTM, and
the Bi-directional LSTM. However, among all of these models, the RNN with the GRU cell performed
the best, and we only include this model in the analysis below.

We also consider three deep learning models.
One is the
traditional fully connected network (FC), and two BNN models.
The first is the standard BNN which has a Gaussian approximate posterior (BNN: Gaus), and the second is our proposed Spike and
Slab Radial (BNN: Spike-Slab Radial).
For these three networks, we ignore the temporal aspect of the data by aggregating all available
features per entry with the `logical or' operator. Since our features are binary,
aggregation corresponds to summarizing the information into the set of events which
occurred during the time period.
In addition, we also considered an approach with a Bayesian Network (i.e., not a
BNN). However, the BN model failed to converge due to the sparsity and high dimensionality of the data.
Furthermore, we also trained many variants of XGBoost~\cite{10.1145/2939672.2939785}, but all of the boosted decision tree models
produced random results. Therefore, we did not include the results for XGBoost below.

\textbf{Parameter settings/hardware.}
All experiments were run on an NVIDIA P100.
 The code was implemented in PyTorch,
 using the Adam optimizer \cite{kingma2014adam} for all models,
 and trained for 400 epochs. The model with the lowest validation loss was selected for evaluation.

\textbf{Model Evaluation.}
In Figure~\ref{fig:roc}, we provide the ROC curves for the proposed model and several baselines.
For our Radial Spike and Slab BNN method and the Gaussian BNN method, we display the distribution of each model's ROC curves, shaded in green,
together with its mean value (e.g., green line).
Figure~\ref{fig:roc_within} shows that, with respect to the distribution of the ROC curves, our method outperforms the other baselines on the validation set, on average. In addition, the  Radial Spike and Slab BNN is able to provide a range of ROC curves which are significantly higher than the other baselines, if we consider the margins of the ROC distribution.
Looking at the columns for the validation set in Table~\ref{tab:sum}, we see that proposed BNN model outperforms the baseline methods, w.r.t. to AUC, accuracy, G-Mean, and other statistics.

Next, we evaluate the ability of our model to preserve the prediction power by applying it to the test data.
While this experiment fairly represents how the model might perform in practice with
application of off-line (batch) models to the new data, we should expect the decrease in performance of the model for the test set compared to the results on the validation set. Indeed, comparing Figure~\ref{fig:roc_outside} to Figure~\ref{fig:roc_within}, we see an overall performance decrease in \textit{all} models.
Despite that, we see from Figure~\ref{fig:roc_outside} that our model still outperforms the baselines on average, particularly in the region of small false positive rates, Figure~\ref{fig:roc_outsidezoom}. Our conclusion is supported by most of the statistics in the `Test Set' columns in Table~\ref{tab:sum}.

\begin{figure*}[!ht]
    \centering
    \includegraphics[width=0.7\textwidth]{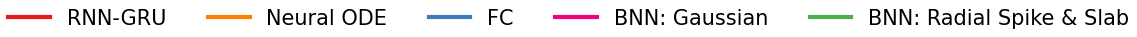}

    \begin{subfigure}{0.30\textwidth}
        \centering
        \includegraphics[width=1\textwidth]{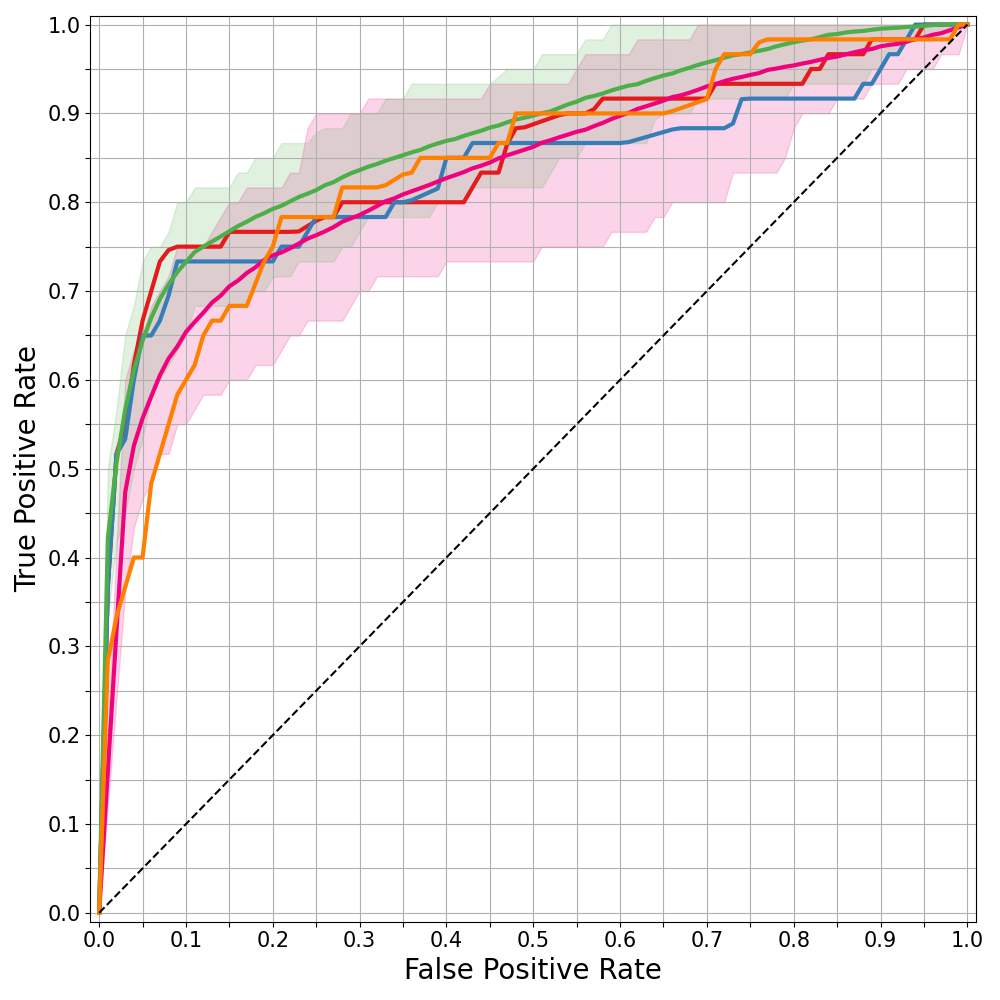}
        \caption{Validation Set}
        \label{fig:roc_within}
    \end{subfigure}
    \begin{subfigure}{0.30\textwidth}
        \centering
        \includegraphics[width=1\textwidth]{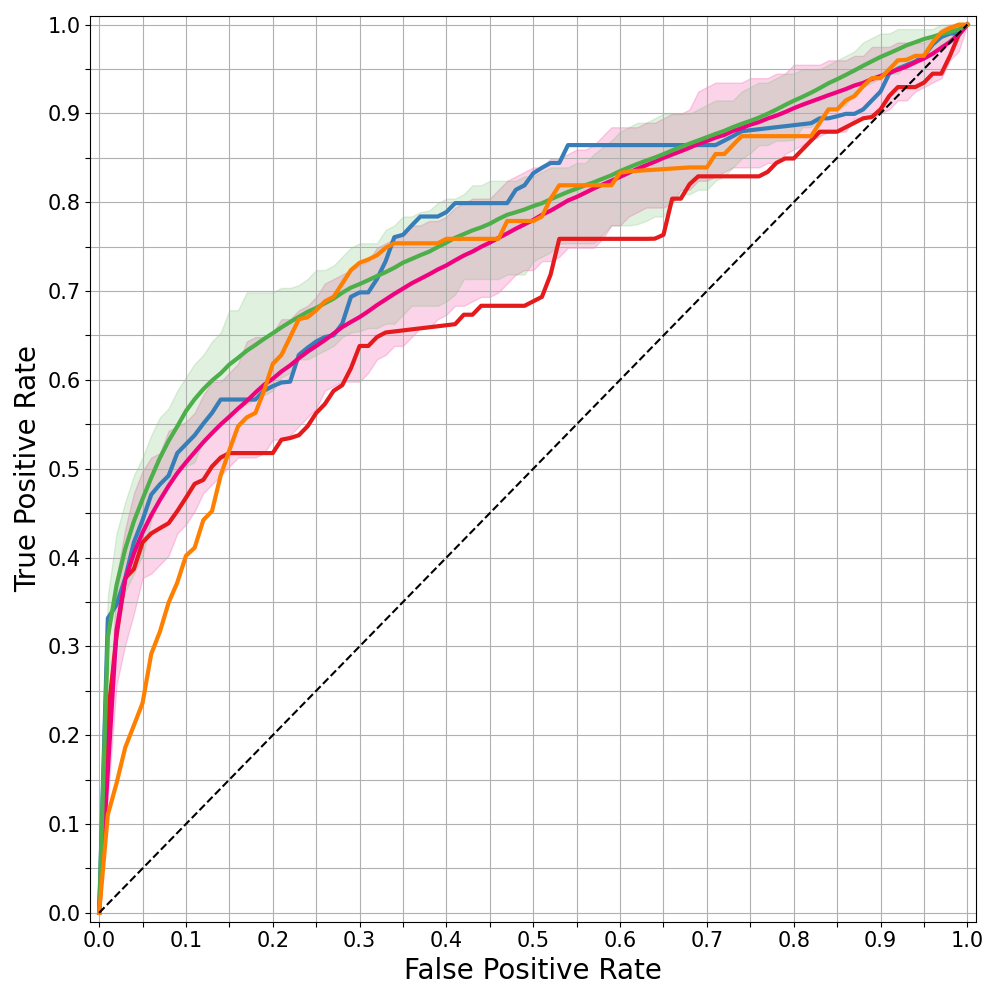}
        \caption{Test Set}
        \label{fig:roc_outside}
    \end{subfigure}
    \begin{subfigure}{0.30\textwidth}
        \centering
        \includegraphics[width=1\textwidth]{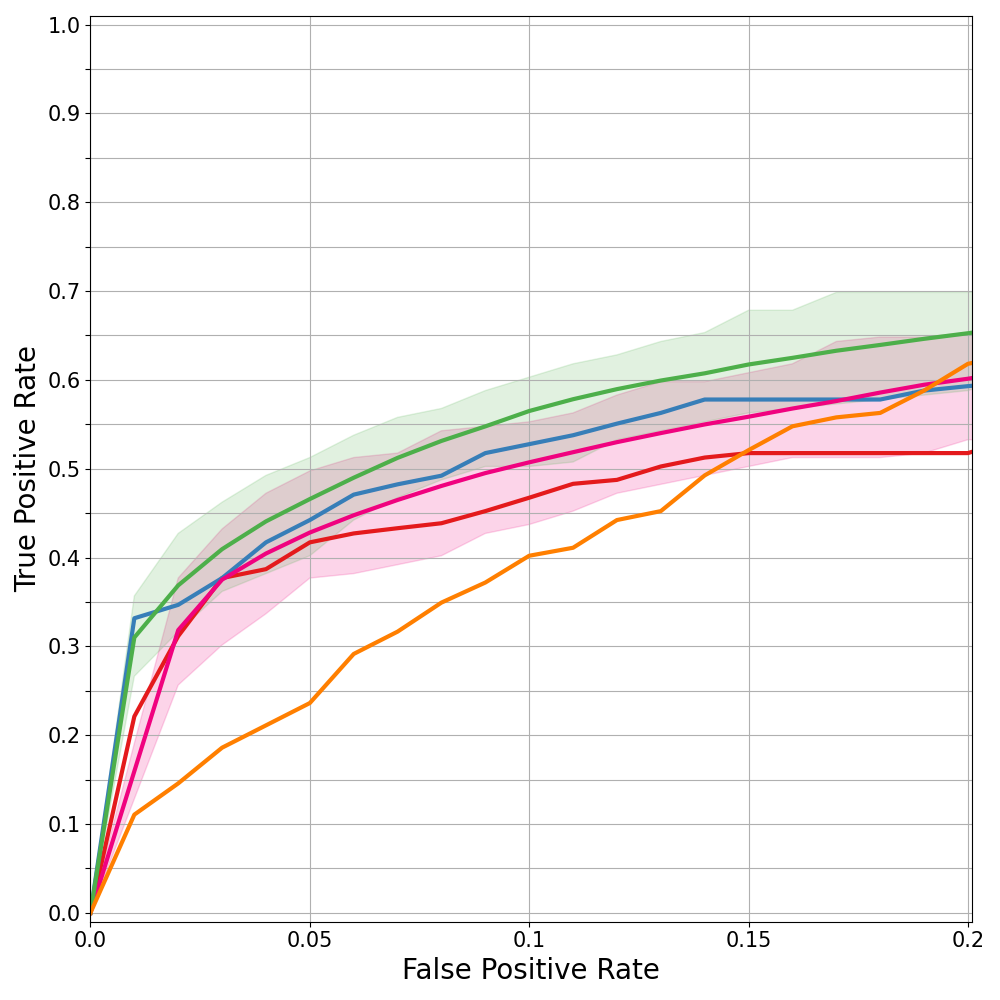}
        \caption{Test Set - Zoomed In}
        \label{fig:roc_outsidezoom}
    \end{subfigure}

    \caption{We present the ROC curves for validation set, and the new data from the future time period of time in the `Test Set'. Because the BNN is a probabilistic model, we show the distribution of the individual ROC curves (green shade) with the mean of this distribution (green line).}
    \label{fig:roc}
\end{figure*}
\begin{table*}[!ht]
    \centering

%
%
%
\resizebox{\textwidth}{!}{%
\begin{tabular}{l|ccccc|ccccc}
\specialrule{1pt}{1pt}{0pt}
\rowcolor{tableheadcolor}
\cellcolor{orange} & \multicolumn{5}{c|}{Validation Set} & \multicolumn{5}{c}{Test Set: Future Time Period} \\
\cline{2-11}
\multirow{-2}{*}{\cellcolor{orange}Statistics} &  RNN-GRU & Neural ODE &   FC & \shortstack{\\BNN: \\Gaussian} & \shortstack{\\BNN:\\ Radial Spike \& Slab} &                      RNN-GRU & Neural ODE &   FC & \shortstack{\\BNN: \\Gaussian} & \shortstack{\\BNN:\\ Radial Spike \& Slab} \\
\toprule
AUC         &                                0.85 & 0.83 & 0.83 &      0.83 &                   {\bf0.87} &                            0.70 & 0.73 & \bf 0.77 &      0.75 &                   {\bf0.77} \\
Sensitivity       &                                0.75 & \bf 0.77 & 0.73 &      0.73 &                   0.73 &                            0.47 & \bf 0.62 & 0.58 &      0.58 &                   0.54 \\
Specificity       &                                0.90 & 0.80 & 0.88 &      0.89 &                   {\bf0.93} &                            0.90 & 0.79 & 0.82 &      0.89 &                   {\bf0.92} \\
Precision         &                                0.06 & 0.03 & 0.05 &      0.05 &                   {\bf0.08} &                            0.09 & 0.06 & 0.06 &      0.10 &                   {\bf0.13} \\
FPR               &                                0.10 & 0.20 & 0.12 &      0.11 &                   \bf 0.07 &                            0.10 & 0.21 & 0.18 &      0.11 &                   \bf 0.08 \\
FNR               &                                0.25 & \bf 0.23 & 0.27 &      0.27 &                   0.27 &                            0.53 & \bf 0.38 & 0.42 &      0.42 &                   0.46 \\
FDR               &                                0.94 & 0.97 & 0.95 &      0.95 &                   \bf 0.92 &                            0.91 & 0.94 & 0.94 &      0.90 &                   \bf 0.87 \\
Accuracy          &                                0.90 & 0.80 & 0.87 &      0.89 &                   {\bf0.93} &                            0.89 & 0.79 & 0.81 &      0.88 &                   {\bf0.91} \\
Balanced Accuracy &                                0.82 & 0.78 & 0.80 &      0.81 &                   {\bf0.83} &                            0.68 & 0.71 & 0.70 &      \bf 0.74 &                   0.73 \\
$F_1$               &                                0.11 & 0.06 & 0.08 &      0.10 &                   {\bf0.14} &                            0.15 & 0.11 & 0.12 &      0.18 &                   {\bf0.21} \\
$F_2$               &                                4.13 & \bf 4.50 & 4.26 &      4.19 &                   3.90 &                            3.37 & \bf 3.93 & 3.85 &      3.43 &                   3.16 \\
G-Mean            &                                0.82 & 0.78 & 0.80 &      0.81 &                   {\bf0.83} &                            0.65 & 0.70 & 0.69 &      \bf 0.72 &                   0.70 \\
\bottomrule
\end{tabular}
}
\caption{The summary statistics are provided for both validation set and the test set, which contains data from the future.  }
\label{tab:sum}
\end{table*}

{\bf Training and Test Times.}
Training the Radial Spike and Slab Bayesian Neural Network in a single Azure-hosted Linux VM with an NVIDIA P100 for 400 epochs required 1 hour, 32 minutes and 53 seconds.
The time required to evaluate the 15,340 samples in the test set was 19 seconds.

\textbf{Feature Importance and Interpretation.}
We would like to understand which TTP features of the attack
are considered to be important by our model when making a prediction whether an
attack is ransomware or not. One method to do this is to investigate
the posterior probabilities for the first layer weights of the BNN.
However, while understanding which TTP features are important based on the BNN's trained weights conceptually makes sense,
we instead follow a more
well-known and established way to interpret the features of a
general neural network, called Integrated Gradients~\cite{sundararajan2017axiomatic}.
In Table~\ref{tab:import}, we demonstrate the subset of features which are the most important for our model to identify
whether an attack is ransomware or some other type of attack based on Integrated Gradients. From the rankings, we find that the
MITRE ATT\&CK features are significantly more important than the rule-based features.

\begin{table}
\centering
\resizebox{0.5\columnwidth}{!}{
\begin{tabular}{c|l}
\specialrule{1pt}{1pt}{0pt}
\rowcolor{tableheadcolor}
Id & Feature representation\\
\toprule
T1059.001 & Command and Scripting Interpreter, Powershell \\
T1105 & Ingress Tool Transfer  \\
T1087 & Account Discovery  \\
Rule & 20ef556c-cc3e-432a-b3ec-1050da7a9fla  \\
T1049 & System Network Connections Discovery  \\
\midrule
T1027.002  & Obfuscated Files or Information: Software Packing\\
T1566.001 & Phishing: Spearphishing Attachment  \\
T1546.001 & Event Triggered Execution: Change Default File Association  \\
T1218.003 & Signed Binary Proxy Execution: CMSTP  \\
T1055.004 & Process Injection: Asynchronous Procedure Call  \\
\bottomrule
\end{tabular}
}
\caption{The Integrated Gradients method produces a score for each of the TTP features which indicates the importance of the feature for predicting whether the attack is ransomware (top) or another type (bottom).}
\label{tab:import}
\end{table} 
\section{Related Work}
\label{sec:related}
Recently, ransomware has become an active research area~\cite{10.1145/3514229,10.1145/3479393}.
Machine learning approaches have  been proposed for the detection of ransomware attacks.
A stacked, variational autoencoder is used to detect ransomware in the industrial IoT (IIoT) setting~\cite{10.1145/3361758.3361763}.
System API calls are used to detect ransomware using Decision Trees,
a K-Nearest Neighbor classifier, and a Random Forest in~\cite{8554938}. Takeuchi et al.~\cite{10.1145/3229710.3229726} also proposed using an SVM to detect ransomware using System API calls.
Agrawal et al.~\cite{8682899} proposed a new attention mechanism on the input vector of an LSTM, an RNN and a GRU to improve the detection of ransomware attacks from API calls.
An ensemble of network traffic classifiers are used to detect network packets and flows for the Locky family of ransomware in~\cite{8674751}. A Bayesian Network was the best performing flow-based classifier in this work while a Random Tree was the best for detecting packets in this work.
HelDroid~\cite{Andronio2015HelDroid_Dissecting} uses natural language processing techniques, along with static and dynamic analysis, to detect ransomware on mobile computing devices.
Adamov and Carlsson~\cite{9225141} use reinforcement learning to simulate ransomware attacks for testing ransomware detectors.  Urooj et al.~\cite{9392548} proposed an online classifier to predict early stage ransomware, but they do not provide any details for the classifier itself.
\section{Limitations and Conclusion}
\label{sec:conc}
In this work, we propose the new Radial Spike and Slab Bayesian Neural Network and demonstrate that it outperforms the standard Bayesian Neural Network and other deep learning methods for the task of detecting ransomware attacks within the general class of all attacks, such as the dropping of commodity malware. The results can provide an early indicator of a potential ransomware attack for analysts to be
able to confirm with additional investigation.

While the model is able to learn to distinguish between ransomware attacks and other attacks, the
ROC curve indicates that it cannot be used by a fully automated system to completely disable computers or block
network access due to a potential ransomware outbreak. However, since these attacks are being diagnosed by analysts, we believe that the model
can alert these analysts that they might have an active ransomware attack on their network.

Finally, given that ransomware attacks are relatively rare compared to the downloading of commodity malware, the amount of
labeled data for these types of attacks is small. The size of our datasets from a production security service reflect this
limitation.
Fortunately, Bayesian computational methods such as Bayesian Neural Networks
can be used for training and inference without overfitting in scenarios where the amount of labeled data is limited.

\bibliography{ref}

\end{document}